\documentclass[12pt]{article}

\usepackage[a4paper,margin=1in]{geometry}
\usepackage[T1]{fontenc}
\usepackage[utf8]{inputenc}
\usepackage{lmodern}
\usepackage{microtype}
\usepackage{amsmath,amssymb,amsthm,mathtools}
\usepackage{booktabs}
\usepackage{enumitem}
\usepackage{setspace}
\usepackage[backend=biber,style=authoryear,maxcitenames=2,maxbibnames=99]{biblatex}
\newtheorem{theorem}{Theorem}[section]
\newtheorem{proposition}[theorem]{Proposition}
\newtheorem{lemma}[theorem]{Lemma}
\newtheorem{corollary}[theorem]{Corollary}

\newtheorem{axiom}{Axiom}[section]

\usepackage{hyperref}
\usepackage[nameinlink,noabbrev]{cleveref}
\hypersetup{colorlinks=true,linkcolor=black,citecolor=black,urlcolor=black}

\crefname{theorem}{Theorem}{Theorems}
\crefname{proposition}{Proposition}{Propositions}
\crefname{lemma}{Lemma}{Lemmas}
\crefname{corollary}{Corollary}{Corollaries}
\crefname{definition}{Definition}{Definitions}
\crefname{axiom}{Axiom}{Axioms}
\crefname{remark}{Remark}{Remarks}
\crefname{example}{Example}{Examples}

\newcommand{\X}{X}
\newcommand{\M}{\mathcal M}
\newcommand{\CO}{C^{O}}
\newcommand{\CF}{C^{F}}
\newcommand{\succO}{\succ^{O}}
\newcommand{\succF}{\succ^{F}}
\newcommand{\succeqN}{\succeq^{N}}
\newcommand{\succN}{\succ^{N}}
\newcommand{\succeqH}{\succeq^{H}}
\newcommand{\succH}{\succ^{H}}
\newcommand{\PNH}{P^{N,H}}
\newcommand{\POH}{P^{O,H}}
\newcommand{\Fset}{\mathfrak F}
\newcommand{\Iset}{\mathcal I}
\newcommand{\InvH}{\operatorname{Inv}_{H}}
\newcommand{\LE}{\operatorname{LE}}
\newcommand{\Max}{\operatorname{Max}}
\newcommand{\tc}{\operatorname{tc}}
\newcommand{\topc}{\operatorname{top}}

\title{Directional Revision under Two-Horizon Deliberation:\\A Revealed-Preference Analysis}
\author{Sinan Ertemel\\[0.25em]
\small Department of Economics, Istanbul Technical University\\
\small Ma\c{c}ka 34367 Istanbul, T\"urkiye\\
\small Corresponding author: \href{mailto:ertemels@itu.edu.tr}{ertemels@itu.edu.tr}\\
\small ORCID: 0000-0003-0089-4641}
\date{}

\begin{document}
\maketitle

\begin{abstract}
We study deterministic choice when the same decision maker is observed before and after a deliberative intervention that makes one of two fixed ordinal consequence dimensions more salient. Both choice modes are path independent. Ordinary choice respects two-dimensional dominance, while any binary reversal under fuller deliberation must favor the alternative that is strictly better on the emphasized dimension. We characterize exactly the admissible ordinary and full-horizon rankings. The characterization yields a protected partial order, sharp pairwise and menu-level identification, a signed restriction on menu-level choice changes, and an acyclicity test for incomplete observations. With an ordered sequence of deliberative modes, choices from any fixed menu form a monotone tradeoff path: every switch improves the emphasized dimension, worsens the other dimension, and an abandoned alternative cannot reappear. When the emphasized consequence order is strict, the admissible rankings have an equivalent Kemeny--Kendall representation. The framework is ordinal: it requires neither a cardinal tradeoff nor lexicographic priority, and it applies whenever deliberation gives greater emphasis to one of two fixed consequence dimensions.
\end{abstract}


\section{Introduction}\label{sec:intro}

Many economically relevant interventions leave the available alternatives and factual information unchanged but alter which already-known consequences are foregrounded in deliberation. A planning prompt may emphasize long-run welfare relative to current benefit; an ethical prompt may foreground moral consequences relative to private payoff; an environmental prompt may make future externalities more salient relative to immediate consumption; and a pro-social prompt may direct attention from self-interest toward effects on others. These settings share a revealed-choice question: if two ordinal consequence dimensions are fixed and an intervention gives one of them greater deliberative emphasis, what restrictions does that direction impose on observed choice? We study this question without specifying a cardinal aggregator, a target utility function, or a rule that resolves every cross-dimensional tradeoff.

We observe two single-valued choice functions of the same individual, called \emph{ordinary choice} and \emph{full-horizon choice}. The maintained comparison holds fixed the menu, factual information, feasibility, and two weak orders over consequences; only the deliberative mode changes. We use the mnemonic labels \emph{Now} ($N$) and \emph{Hereafter} ($H$) for these two consequence rankings. The terminology is partly motivated by an Islamic-economics application in which worldly and afterlife consequences are explicitly distinguished; \Cref{sec:motivation} gives a slightly fuller motivation. This application is only one interpretation of the general structure, and none of the formal results depends on it. The $N$- and $H$-orders are not two overall preferences or two selves. They are consequence rankings common to both modes, whereas the overall rankings revealed by ordinary and full-horizon choice may differ.

Several recent papers study nearby but distinct questions. \Textcite{Spohn2025Reflective} develops reflective choice for changing decision situations and allows epistemic and preference change; \textcite{GruneYanoff2025Constructed} studies when preference construction should be completed relative to action; and \textcite{Petersen2026Dynamic} studies dynamic consistency with incomplete but stable preferences. \Textcite{BreigFeldman2026RP} derives revealed-preference tests for parametric models in a risk-elicitation environment. Especially close in spirit, \textcite{BreigFeldman2024Revisions} studies choice revisions without informational change and interprets revised choices as mistakes. Here a revision is not presumed to correct an error, and we do not model decision-tree dynamics, preference construction, or a target parametric utility. Instead, we hold the decision environment and two consequence rankings fixed and characterize the cross-mode reversals consistent with a directional deliberative intervention.

The model imposes three restrictions. First, each choice mode is path independent and hence rationalized by a unique strict ranking. Second, ordinary choice respects \emph{two-horizon dominance}: an alternative that is weakly better on both consequence dimensions and strictly better on at least one cannot lose a binary ordinary-choice comparison. Third, \emph{Hereafter-Directed Revision} requires any binary reversal induced by full-horizon deliberation to favor the alternative that is strictly better under the $H$-ranking. The converse is deliberately absent: an $H$-advantage need not overturn every $N$-advantage.

The central characterization identifies a protected partial order generated jointly by ordinary revealed preference and the $H$-ranking. Full-horizon rankings are exactly its linear extensions. This yields sharp pairwise and menu-level identification and a signed choice implication: whenever ordinary and full-horizon choices differ on a menu, the full-horizon winner is strictly better under $H$, while the ordinary winner is strictly better under $N$. When $H$ is strict, the same restriction has an equivalent Kemeny--Kendall representation. With an ordered sequence of deliberative modes, the model further implies monotone menu trajectories: every actual switch improves $H$ and worsens $N$, an abandoned alternative cannot reappear, and the number of switches is bounded by the height of the induced tradeoff order.

The analysis also connects to broader revealed-choice and preference-revision literatures. Models such as \textcite{GulPesendorfer2001}, \textcite{BernheimRangel2009}, and \textcite{SilvaLeanderSeth2017} separate observed choice from other preference or welfare objects, while \textcite{Freund2005} and \textcite{HaretWallner2022} study formal preference revision under new information or authoritative input. Our two consequence rankings instead remain fixed across modes, and neither is assumed to be the individual's overall preference. The tie-free case connects the resulting behavioral restriction to standard betweenness geometry for linear orders \parencite{KemenySnell1962,Grandmont1978,BossertSprumont2014}.

Several tools used below are classical. Path independence supplies the within-mode order representation \parencite{Plott1973}; finite order-extension results supply the linear-extension and acyclicity arguments \parencite{DushnikMiller1941}; and Kemeny--Kendall distance supplies the tie-free geometry \parencite{Kendall1938,KemenySnell1962,Grandmont1978}. Our contribution is not new order theory, but the choice-theoretic structure that generates the protected relation and links it exactly to directional deliberation.

The rest of the paper is organized as follows. \Cref{sec:motivation} gives an illustrative benchmark. \Cref{sec:model} introduces the choice environment and axioms. \Cref{sec:char} proves the main characterization. \Cref{sec:id} develops identification and observable implications. \Cref{sec:geometry} studies the tie-free case and repeated deliberation. \Cref{sec:data} gives rationalizability tests for incomplete data and additional pairwise restrictions. \Cref{sec:conclusion} concludes.

\section{An illustrative benchmark}\label{sec:motivation}

A simple example makes the distinction between dominance and cross-horizon tradeoffs concrete. We use an Islamic-economic illustration because it gives a natural interpretation to the labels Now and Hereafter, while remaining only one application of the general framework. Among several Qur'anic passages that distinguish worldly and afterlife considerations, we focus on two passages especially suited to the present economic reading. Qur'an 2:219--220 explicitly links reflection to ``this world and the Hereafter,'' while Qur'an 28:77 combines seeking the Hereafter with not neglecting one's share of the world; see \textcite{saheehinternational1997} for the English translation. We use these sources only to motivate a model in which both horizons may matter without assuming that one always overrides the other. This is a limited economic reading, not an exegesis or a comprehensive account of Islamic decision making.

Consider a household that has met its obligations and has limited discretionary income remaining. It can either replace a functioning car with a newer one that provides additional comfort while maintaining a relatively lower voluntary charitable contribution, or keep the current car and make a higher contribution. Suppose the available surplus is insufficient to do both.

Holding other relevant considerations fixed, let the first alternative have the better Now consequence and the second the better Hereafter consequence. Neither dominates the other. By contrast, if one feasible alternative were weakly better on both dimensions and strictly better on at least one, the comparison would involve no cross-horizon sacrifice and would be a dominance comparison. The example is deliberately stylized; it does not claim that a car purchase has no moral or afterlife relevance or that a larger donation is always the uniquely appropriate choice.

The model does not determine which side of the tradeoff must be chosen. Fuller deliberation may leave the ordinary choice unchanged or may shift choice toward the alternative with the better Hereafter consequence. What it rules out is attributing a switch in the opposite direction to an intervention whose only maintained change is fuller attention to the two consequence dimensions. The restriction is therefore directional rather than lexicographic: it constrains revisions when they occur, not the level of choice in every tradeoff.

The next section replaces this example with arbitrary alternatives, menus, and ordinal Now and Hereafter rankings.

\section{Choice environment and axioms}\label{sec:model}

Let $\X$ be a finite set of alternatives, with $|\X|\ge2$, and let
\[
 \M=2^{\X}\setminus\{\varnothing\}
\]
be the set of all nonempty menus. A single-valued choice function $C:\M\to\X$ satisfies $C(S)\in S$ for every $S\in\M$.

There are two consequence rankings on $\X$. The \emph{Now ranking} $\succeqN$ and the \emph{Hereafter ranking} $\succeqH$ are weak orders: complete and transitive binary relations, possibly with ties. Their strict parts are denoted by $\succN$ and $\succH$. Thus, for example,
\[
 x\succH y
 \quad\Longleftrightarrow\quad
 x\succeqH y\ \text{and not }y\succeqH x.
\]
The formal analysis uses only these ordinal rankings.

Define strict two-horizon dominance by
\[
 x\PNH y
 \quad\Longleftrightarrow\quad
 x\succeqN y,\quad x\succeqH y,
 \quad\text{and}\quad
 [x\succN y\text{ or }x\succH y].
\]
The relation $\PNH$ is a strict partial order. Irreflexivity is immediate. For transitivity, suppose $x\PNH y$ and $y\PNH z$. Transitivity of $\succeqN$ and $\succeqH$ gives $x\succeqN z$ and $x\succeqH z$. At least one of $x\succN y$ or $x\succH y$ holds. If $x\succN y$, then $y\succeqN z$ implies $x\succN z$; if $x\succH y$, then $y\succeqH z$ implies $x\succH z$. Thus at least one strict improvement survives the composition, and $x\PNH z$.

We observe two choice functions. The ordinary choice function is $\CO$ and the full-horizon choice function is $\CF$. Both are defined on $\M$. They represent the same decision maker in two deliberative modes. The maintained comparison across modes holds fixed the menu, factual information, feasibility, and the two consequence rankings; only the deliberative mode changes.

Our first axiom imposes standard within-mode consistency.

\begin{axiom}[Path Independence]\label{ax:pi}
For $K\in\{O,F\}$ and all $S,T\in\M$,
\[
 C^K(S\cup T)
 =C^K\bigl(\{C^K(S),C^K(T)\}\bigr).
\]
\end{axiom}

On the present finite and single-valued universal domain, Path Independence yields the familiar rational-choice representation associated with \textcite{Plott1973}. We record the specialization used below.

\begin{lemma}[Choice representation; Plott]\label{lem:choice-rep}
A single-valued choice function $C:\M\to\X$ is path independent if and only if there exists a unique strict linear order $\succ$ such that, for every $S\in\M$, $C(S)$ is the $\succ$-best element of $S$.
\end{lemma}

\begin{proof}
The ``if'' direction is immediate. For the converse, define $x\succ y$ for distinct $x,y$ whenever $C(\{x,y\})=x$. Single-valuedness makes $\succ$ complete and asymmetric. Suppose $x\succ y$ and $y\succ z$. Path Independence applied to $\{x,y\}$ and $\{y,z\}$ gives $C(\{x,y,z\})=x$. If $z\succ x$, Path Independence applied to $\{x,z\}$ and $\{y\}$ instead gives $C(\{x,y,z\})=y$, a contradiction. Thus $x\succ z$.

Now let $x=C(S)$. If some $y\in S\setminus\{x\}$ satisfied $y\succ x$, then Path Independence applied to $S$ and $\{x,y\}$ would imply $C(S)=y$, a contradiction. Hence $x$ is above every other element of $S$. Binary choices uniquely determine the order.
\end{proof}

Under Path Independence, write $\succO$ and $\succF$ for the strict linear orders rationalizing $\CO$ and $\CF$.

The first substantive axiom concerns dominance.

\begin{axiom}[Two-Horizon Dominance]\label{ax:thd}
For distinct $x,y\in\X$, if $x\PNH y$, then
\[
 \CO(\{x,y\})=x.
\]
\end{axiom}

The axiom is silent when the two horizons conflict; it excludes only an ordinary binary choice of a dominated alternative.

The second substantive axiom concerns changes between deliberative modes.

\begin{axiom}[Hereafter-Directed Revision]\label{ax:adr}
For distinct $x,y\in\X$, if
\[
 \CF(\{x,y\})\neq\CO(\{x,y\}),
\]
then
\[
 \CF(\{x,y\})\succH\CO(\{x,y\}).
\]
\end{axiom}

The axiom is one-sided: $x\succH y$ does not require full-horizon choice to select $x$. It constrains only reversals; in particular, a pair tied under $\succeqH$ cannot reverse across modes.

\section{Characterization}\label{sec:char}

Fix $\succO$ and $\succeqH$. Call an ordinary comparison \emph{protected} when its ordinary winner is weakly no worse in the Hereafter. Define the protected ordinary relation $\POH$ by
\[
 x\POH y
 \quad\Longleftrightarrow\quad
 x\succO y\ \text{and}\ x\succeqH y.
\]
The relation $\POH$ is a strict partial order. It is irreflexive because $\succO$ is irreflexive. If $x\POH y$ and $y\POH z$, then $x\succO z$ by transitivity of $\succO$ and $x\succeqH z$ by transitivity of $\succeqH$, so $x\POH z$.

For a strict partial order $P$, let $\LE(P)$ denote the set of strict linear orders extending $P$.

\begin{lemma}[Directional-revision characterization]\label{lem:adr-char}
Suppose Path Independence holds. For fixed $\succO$ and $\succeqH$, Hereafter-Directed Revision holds if and only if
\[
 \succF\in\LE(\POH).
\]
\end{lemma}

\begin{proof}
Suppose Hereafter-Directed Revision holds and take $x\POH y$. Then $x\succO y$, so Path Independence and Lemma~\ref{lem:choice-rep} imply $\CO(\{x,y\})=x$. If, contrary to the claim, $y\succF x$, then $\CF(\{x,y\})=y$ and the binary choice is reversed. Hereafter-Directed Revision would require $y\succH x$, contradicting $x\succeqH y$. Hence $x\succF y$, and $\succF$ extends $\POH$.

Conversely, suppose $\succF$ extends $\POH$ and a binary choice reverses. Relabeling the pair if necessary, Lemma~\ref{lem:choice-rep} gives $x\succO y$ and $y\succF x$. If $x\succeqH y$, then $x\POH y$, so extension of $\POH$ would give $x\succF y$, a contradiction. Therefore $x\not\succeqH y$. Completeness of $\succeqH$ yields $y\succeqH x$, and together with $x\not\succeqH y$ this is exactly $y\succH x$. Thus Hereafter-Directed Revision holds.
\end{proof}

\begin{theorem}[Two-horizon characterization]\label{thm:main}
Suppose Path Independence holds. Two-Horizon Dominance and Hereafter-Directed Revision hold if and only if
\[
 \succO\in\LE(\PNH)
 \qquad\text{and}\qquad
 \succF\in\LE(\POH).
\]
Moreover, whenever Two-Horizon Dominance holds,
\[
 \PNH\subseteq\POH.
\]
Consequently every admissible full-horizon ranking also respects two-horizon dominance.
\end{theorem}

\begin{proof}
Under Path Independence, Two-Horizon Dominance is equivalent to requiring $\succO$ to extend $\PNH$. By Lemma~\ref{lem:adr-char}, Hereafter-Directed Revision is equivalent to requiring $\succF$ to extend $\POH$. This proves the characterization.

For the inclusion, let $x\PNH y$. By Path Independence and Two-Horizon Dominance, $\CO(\{x,y\})=x$ implies $x\succO y$, while the definition of $\PNH$ gives $x\succeqH y$. Hence $x\POH y$.
\end{proof}

The theorem identifies exactly which ordinary comparisons full-horizon choice must preserve. Comparisons aligned with the Hereafter ranking are protected; ordinary comparisons opposed by a strict Hereafter comparison may survive or reverse, subject to transitivity.

We next use two standard facts about finite partial orders. The first is the classical Dushnik--Miller representation of a partial order as the intersection of its linear extensions; the second is the corresponding extension fact that any maximal element of a subset can be placed first within that subset in some linear extension \parencite{DushnikMiller1941}. We state the two facts in the form needed for identification.

\begin{lemma}[Standard linear-extension facts]\label{lem:extensions}
Let $P$ be a strict partial order on a finite set.
\begin{enumerate}[label=(\roman*),leftmargin=*]
 \item $P$ is the intersection of its strict linear extensions:
 \[
 P=\bigcap_{L\in\LE(P)}L.
 \]
 \item For a nonempty menu $S$ and $x\in S$, there exists $L\in\LE(P)$ that ranks $x$ above every other element of $S$ if and only if $x$ is $P$-maximal in $S$.
\end{enumerate}
\end{lemma}

\begin{proof}
Every linear extension contains $P$, so only the reverse inclusion in (i) requires proof. Fix distinct $x,y$ with $xPy$ false. If $yPx$, then every linear extension ranks $y$ above $x$, so $(x,y)$ is not in the intersection. If $x$ and $y$ are incomparable, augment $P$ by the edge $(y,x)$. This augmented relation is acyclic: a directed cycle would contain the new edge and therefore a $P$-path from $x$ back to $y$, which by transitivity of $P$ would imply $xPy$, contrary to assumption. Its transitive closure is therefore a strict partial order and, by the finite order-extension theorem, has a strict linear extension. That extension ranks $y$ above $x$. Hence $x$ is above $y$ in every linear extension if and only if $xPy$, proving (i).

For (ii), necessity is immediate: if some $y\in S$ satisfies $yPx$, every extension ranks $y$ above $x$. For sufficiency, suppose $x$ is $P$-maximal in $S$ and augment $P$ by the edges $(x,y)$ for all $y\in S\setminus\{x\}$. The augmented relation is acyclic. Indeed, any directed cycle must use at least one added edge, say $(x,y)$; after traversing that edge the cycle must return from $y$ to $x$ using only edges of $P$ before another added edge can be used, yielding a $P$-path from $y$ to $x$. Transitivity would imply $yPx$, contradicting maximality. The transitive closure of the augmented relation is thus a strict partial order and has a strict linear extension. In that extension, $x$ is above every other element of $S$.
\end{proof}

\section{Identification and observable implications}\label{sec:id}

Fix $\succO$ and $\succeqH$. Let
\[
 \Fset(\succO;\succeqH)=\LE(\POH)
\]
be the set of full-horizon rankings compatible with Hereafter-Directed Revision. Two-Horizon Dominance restricts which ordinary rankings are admissible in the full model, but conditional on $\succO$ the identification results below depend only on the directional-revision restriction.

\subsection{Pairwise and menu-level identification}

\begin{corollary}[Robust full-horizon relation]\label{thm:robust}
For distinct $x,y\in\X$,
\[
 x\POH y
 \quad\Longleftrightarrow\quad
 x\succF y
 \text{ for every }\succF\in\Fset(\succO;\succeqH).
\]
Equivalently,
\[
 \POH
 =\bigcap_{\succF\in\Fset(\succO;\succeqH)}\succF.
\]
\end{corollary}

\begin{proof}
This is the Dushnik--Miller intersection result in Lemma~\ref{lem:extensions}(i), specialized to the model-implied protected relation $\POH$.
\end{proof}

For a strict partial order $P$, let $\Max(S,P)$ denote the set of $P$-maximal elements of a menu $S$. For a strict linear order $\succ$, let $\topc_{\succ}(S)$ denote its top element in $S$. Define the identified full-horizon choice set by
\[
 \Iset_F(S\mid\succO,\succeqH)
 =\{\topc_{\succF}(S):\succF\in\Fset(\succO;\succeqH)\}.
\]

\begin{corollary}[Sharp menu identification]\label{thm:menu}
For every nonempty menu $S$,
\[
 \Iset_F(S\mid\succO,\succeqH)
 =\Max(S,\POH).
\]
Thus full-horizon choice is point identified on $S$ if and only if $S$ has a unique $\POH$-maximal element.
\end{corollary}

\begin{proof}
By Lemma~\ref{lem:adr-char} and the definition of $\Fset(\succO;\succeqH)$, the admissible full-horizon rankings are exactly the linear extensions of $\POH$. The claim follows from Lemma~\ref{lem:extensions}(ii).
\end{proof}

\begin{corollary}[Global point identification]\label{cor:unique}
The following statements are equivalent:
\begin{enumerate}[label=(\roman*),leftmargin=*]
 \item $\Fset(\succO;\succeqH)$ is a singleton;
 \item $\POH$ is a strict linear order;
 \item whenever $x\succO y$, one has $x\succeqH y$.
\end{enumerate}
If $\succeqH$ has no ties, these conditions are equivalent to $\succO=\succH$.
\end{corollary}

\begin{proof}
A finite strict partial order has a unique linear extension if and only if it is already linear. Since $\POH\subseteq\succO$, the relation $\POH$ is linear exactly when every ordinary comparison has a winner weakly no worse under $\succeqH$. If $\succeqH$ has no ties, two strict linear orders satisfy this condition exactly when they coincide.
\end{proof}

Thus the model leaves unresolved precisely those tradeoffs that the maintained assumptions do not discipline.

\subsection{The directional choice wedge}

The model also signs every observed disagreement between the two modes.

\begin{theorem}[Directional choice wedge]\label{thm:wedge}
Suppose \cref{ax:pi,ax:thd,ax:adr} hold. For a menu $S$, let
\[
 b=\CO(S),\qquad f=\CF(S).
\]
If $b\neq f$, then
\[
 f\succH b
 \qquad\text{and}\qquad
 b\succN f.
\]
Thus disagreement can occur only on a genuine cross-horizon tradeoff: the ordinary winner is Now-superior and the full-horizon winner is Hereafter-superior.
\end{theorem}

\begin{proof}
Since $b$ is the $\succO$-best element of $S$ and $f\in S\setminus\{b\}$, we have $b\succO f$. Likewise, because $f$ is the $\succF$-best element of $S$, $f\succF b$. Lemma~\ref{lem:choice-rep} therefore implies $\CO(\{b,f\})=b$ and $\CF(\{b,f\})=f$. The binary comparison is reversed, so Hereafter-Directed Revision gives $f\succH b$.

Suppose, toward a contradiction, that $f\succeqN b$. Together with $f\succH b$, this gives $f\PNH b$. Two-Horizon Dominance then requires $\CO(\{b,f\})=f$, equivalently $f\succO b$, contradicting $b\succO f$. Hence $f\not\succeqN b$. Completeness of $\succeqN$ gives $b\succeqN f$, and the failure of $f\succeqN b$ makes the comparison strict: $b\succN f$.
\end{proof}

The result is local to the selected alternatives: it signs a change in choice without implying that either mode globally maximizes one consequence dimension.

\section{Tie-free geometry and ordered deliberation}\label{sec:geometry}

\subsection{Tie-free geometry}

Assume in this subsection that the Hereafter ranking has no ties, so $\succH$ is a strict linear order. The characterization then has a particularly transparent metric form.

For two strict linear orders $R$ and $R'$, let $d_K(R,R')$ be their Kendall distance, the number of unordered pairs on which the two orders disagree. A strict linear order $R$ is \emph{Kemeny-between} $R_0$ and $R_1$ when
\[
 d_K(R_0,R_1)=d_K(R_0,R)+d_K(R,R_1).
\]
This is the usual geodesic notion of betweenness generated by pairwise disagreement distance. For strict rankings it coincides with the familiar intermediate-order requirement that the intermediate ranking preserve every pair on which the two endpoint rankings agree \parencite{KemenySnell1962,Grandmont1978,BossertSprumont2014}.

\begin{corollary}[Kemeny representation]\label{cor:kemeny}
Suppose Path Independence holds and $\succH$ is a strict linear order. Then Hereafter-Directed Revision holds if and only if $\succF$ is Kemeny-between $\succO$ and $\succH$:
\[
 d_K(\succO,\succH)
 =d_K(\succO,\succF)+d_K(\succF,\succH).
\]
\end{corollary}

\begin{proof}
When $\succH$ is strict, for distinct alternatives $x\succeqH y$ is equivalent to $x\succH y$, so $\POH=\succO\cap\succH$. By Lemma~\ref{lem:adr-char}, Hereafter-Directed Revision therefore holds exactly when $\succF$ preserves every pair on which the endpoint orders $\succO$ and $\succH$ agree.

Now evaluate the distance identity pair by pair. If $\succO$ and $\succH$ agree on an unordered pair, preserving that comparison contributes zero to all three relevant disagreement counts. If the endpoint orders disagree on the pair, the strict order $\succF$ must agree with exactly one endpoint, so that pair contributes one to $d_K(\succO,\succH)$ and exactly one in total to $d_K(\succO,\succF)+d_K(\succF,\succH)$. Hence preservation of every endpoint-agreement pair implies the displayed equality. Conversely, if the equality holds and $\succF$ reversed even one pair on which the endpoints agree, that pair would contribute zero to the left-hand side but two to the right-hand side, making equality impossible. Thus the distance equality is equivalent to extension of $\succO\cap\succH$, and hence to Hereafter-Directed Revision.
\end{proof}

The geometry is standard; the model-specific point is that Hereafter-Directed Revision places the observed full-horizon ranking in this interval even though $\succH$ ranks only one consequence dimension rather than serving as a target overall preference.

There is an equivalent inversion formulation. Label the alternatives $x_1,\ldots,x_n$ so that
\[
 x_n\succH x_{n-1}\succH\cdots\succH x_1.
\]
For a strict linear order $R$, define its set of Hereafter inversions by
\[
 \InvH(R)
 =\{(i,j):i<j\text{ and }x_i R x_j\}.
\]
An inversion places the Hereafter-worse alternative above the Hereafter-better one.

\begin{corollary}[Inversion characterization]\label{cor:inv}
Suppose Path Independence holds and $\succH$ is strict. Then Hereafter-Directed Revision holds if and only if
\[
 \InvH(\succF)\subseteq\InvH(\succO).
\]
Thus full-horizon deliberation may remove Hereafter inversions of ordinary preference but cannot create new ones.
\end{corollary}

\begin{proof}
Fix a pair $x_i,x_j$ with $i<j$, so $x_j\succH x_i$. If $(i,j)\in\InvH(\succF)\setminus\InvH(\succO)$, then ordinary ranking places the Hereafter-better $x_j$ above $x_i$, while full-horizon ranking reverses the pair and places $x_i$ above $x_j$. The resulting binary reversal is Hereafter-worsening, so Hereafter-Directed Revision fails. Thus the axiom implies $\InvH(\succF)\subseteq\InvH(\succO)$.

Conversely, suppose the inclusion holds and a pair reverses from its ordinary to its full-horizon orientation. If the full-horizon winner were Hereafter-worse than the ordinary winner, that pair would be a Hereafter inversion under $\succF$ but not under $\succO$, contradicting the inclusion. Since $\succH$ is strict, every reversal must therefore move toward the Hereafter-better alternative, which is Hereafter-Directed Revision.
\end{proof}

Define the revision count
\[
 \Delta_H(\succO,\succF)
 =\bigl|\InvH(\succO)\setminus\InvH(\succF)\bigr|.
\]
It counts the ordinary Hereafter inversions removed under full-horizon deliberation.

\begin{proposition}[Monotone Kendall revision]\label{prop:distance}
Suppose Path Independence and Hereafter-Directed Revision hold and $\succH$ is a strict linear order. Then
\[
 \Delta_H(\succO,\succF)=d_K(\succO,\succF).
\]
Moreover, this number is the minimum number of adjacent swaps required to transform $\succO$ into $\succF$. The minimum-adjacent-swap interpretation is the standard Kendall-distance result \parencite{Kendall1938,KemenySnell1962}; the model adds that every swap in any shortest transformation moves the Hereafter-better of the two adjacent alternatives upward.
\end{proposition}

\begin{proof}
By Corollary~\ref{cor:inv}, Hereafter-Directed Revision implies $\InvH(\succF)\subseteq\InvH(\succO)$. Hence every pair on which $\succO$ and $\succF$ disagree is an inversion of $\succO$ removed in $\succF$, and no pair changes in the opposite direction. Hence the disagreement set is exactly $\InvH(\succO)\setminus\InvH(\succF)$, proving the equality with Kendall distance.

The minimum number of adjacent swaps between two linear orders equals their Kendall distance. Consider any shortest sequence. Each adjacent swap changes the number of Hereafter inversions by one. If $m_-$ swaps reduce the inversion count and $m_+$ increase it, then
\[
 m_-+m_+=d_K(\succO,\succF)
\]
and
\[
 m_- - m_+
 =|\InvH(\succO)|-|\InvH(\succF)|
 =d_K(\succO,\succF).
\]
Thus $m_+=0$. Every swap in a shortest sequence therefore reduces the Hereafter-inversion count. For an adjacent pair, reducing that count is equivalent to moving the Hereafter-better alternative above the Hereafter-worse one.
\end{proof}

The metric remains ordinal: $\Delta_H$ counts revised pairwise comparisons, not the intensity of consequence differences.

\subsection{More than two deliberative modes}

The same logic extends to an ordered sequence of deliberative contexts. Let $r=0,1,\ldots,L$ index contexts, let $C^r$ be a single-valued path-independent choice function, and let $\succ^r$ be its rationalizing order. The consequence rankings $\succeqN$ and $\succeqH$ are held fixed across contexts. A higher index represents a context in which the Hereafter dimension is made at least as salient as before.

\begin{axiom}[Stepwise Hereafter-Directed Revision]\label{ax:step}
For every $r<L$ and every distinct $x,y$, if
\[
 C^{r+1}(\{x,y\})\neq C^r(\{x,y\}),
\]
then
\[
 C^{r+1}(\{x,y\})\succH C^r(\{x,y\}).
\]
\end{axiom}

For distinct alternatives, define the \emph{directional tradeoff relation} $\rhd$ by
\[
 x\rhd y
 \quad\Longleftrightarrow\quad
 x\succH y\ \text{and}\ y\succN x.
\]
Thus $x\rhd y$ means that moving from $y$ to $x$ improves the Hereafter dimension while sacrificing the Now dimension. Because the strict parts of weak orders are transitive, $\rhd$ is a strict partial order. For a menu $S$, let $h_{\rhd}(S)$ denote the maximum number of alternatives in a $\rhd$-chain contained in $S$.

\begin{theorem}[Monotone menu trajectories]\label{thm:trajectory}
Suppose every $C^r$ is path independent, $C^0$ satisfies Two-Horizon Dominance, and Stepwise Hereafter-Directed Revision holds. Then:
\begin{enumerate}[label=(\roman*),leftmargin=*]
 \item every $C^r$ satisfies Two-Horizon Dominance;
 \item for every menu $S$ and every $0\le r<s\le L$, either
 \[
 C^r(S)=C^s(S),
 \]
 in which case the choice is the same at every intermediate context $q=r,\ldots,s$, or
 \[
 C^s(S)\succH C^r(S)
 \qquad\text{and}\qquad
 C^r(S)\succN C^s(S);
 \]
 \item after consecutive repetitions are deleted, the sequence
 \[
 C^0(S),C^1(S),\ldots,C^L(S)
 \]
 is a strict $\rhd$-chain. Consequently an alternative that is abandoned on a menu can never be selected again at a later context, and the number of choice changes on $S$ is at most
 \[
 h_{\rhd}(S)-1\le |S|-1.
 \]
\end{enumerate}
\end{theorem}

\begin{proof}
We first prove (i) by induction on $r$. The claim holds at $r=0$ by assumption. Suppose $C^r$ satisfies Two-Horizon Dominance and let $x\PNH y$. Then $C^r(\{x,y\})=x$ and, by definition of $\PNH$, $x\succeqH y$. If $C^{r+1}(\{x,y\})=y$, Stepwise Hereafter-Directed Revision would require $y\succH x$, contradicting $x\succeqH y$. Hence $C^{r+1}(\{x,y\})=x$, so $C^{r+1}$ also satisfies Two-Horizon Dominance.

For (ii), first consider adjacent contexts $r$ and $r+1$. Let $b=C^r(S)$ and $f=C^{r+1}(S)$, and suppose $b\neq f$. Path Independence implies $b\succ^r f$ and $f\succ^{r+1}b$, so the binary choice between $b$ and $f$ reverses. Stepwise Hereafter-Directed Revision gives $f\succH b$. If $f\succeqN b$, then $f\PNH b$, and part (i) would require $f\succ^r b$, contradicting $b\succ^r f$. Hence $b\succN f$. Thus every actual adjacent change satisfies $f\rhd b$.

Now fix $r<s$ and consider the choices from $S$ across contexts $r,\ldots,s$. Delete consecutive repetitions. Each remaining adjacent pair is related by $\rhd$. Transitivity of $\rhd$ therefore implies that, if the endpoints differ, $C^s(S)\rhd C^r(S)$, which is exactly the pair of strict comparisons displayed in (ii). If the endpoints coincide but some intermediate change occurred, transitivity would imply $C^r(S)\rhd C^r(S)$, contradicting irreflexivity. Hence equal endpoints imply no intermediate change.

Part (iii) follows from the same argument: the distinct choices form a strict $\rhd$-chain, so no alternative can reappear after it is left. The number of distinct selected alternatives is therefore at most $h_{\rhd}(S)$, and the number of changes is at most $h_{\rhd}(S)-1$.
\end{proof}

This menu-level implication does not require a tie-free Hereafter ranking. The sharper bound $h_{\rhd}(S)-1$ can be substantially smaller than $|S|-1$ when few alternatives are ordered oppositely by the two consequence rankings.

When the Hereafter ranking is strict, the same sequence also has a ranking-level representation in terms of nested inversion sets.

\begin{corollary}[Nested-inversion characterization]\label{thm:chain}
Suppose every $C^r$ is path independent and $\succH$ is strict. Stepwise Hereafter-Directed Revision holds if and only if
\[
 \InvH(\succ^{L})\subseteq\cdots\subseteq
 \InvH(\succ^{1})\subseteq\InvH(\succ^{0}).
\]
Hence a pair once revised in the Hereafter-improving direction cannot subsequently reverse back under the same maintained sequence of deliberative interventions.
\end{corollary}

\begin{proof}
Each $C^r$ is path independent by construction, so Corollary~\ref{cor:inv} applies to every adjacent pair $(\succ^r,\succ^{r+1})$. Stepwise Hereafter-Directed Revision is therefore equivalent, for each $r<L$, to $\InvH(\succ^{r+1})\subseteq\InvH(\succ^r)$. Collecting these adjacent inclusions gives the displayed chain, and the chain in turn implies every adjacent inclusion and hence the stepwise axiom.
\end{proof}

Thus menu choices are monotone in the tradeoff order even with ties; strict Hereafter rankings additionally imply monotone inversion sets.

\section{Additional restrictions and incomplete data}\label{sec:data}

The linear-extension representation also makes additional pairwise information transparent. We use the standard finite order-extension principle: an acyclic directed relation has an acyclic transitive closure, which is a strict partial order, and every finite strict partial order has a strict linear extension \parencite{DushnikMiller1941}. The model-specific step is that the additional comparisons must be combined with the protected relation. Let $R$ be any directed relation of additional pairwise requirements on distinct alternatives: $xRy$ means that the additional information requires $x$ to be ranked above $y$. Define
\[
 P_R=\tc(\POH\cup R),
\]
where $\tc$ denotes transitive closure.

\begin{proposition}[Pairwise augmentation]\label{thm:augment}
Fix $\succO$ and $\succeqH$. There exists a full-horizon ranking compatible with Hereafter-Directed Revision and all comparisons in $R$ if and only if $\POH\cup R$ is acyclic. When this condition holds, the admissible rankings are exactly
\[
 \LE(P_R),
\]
and the sharp identified choice set on any menu $S$ is
\[
 \Max(S,P_R).
\]
If $R\subseteq\widetilde R$ and both augmented relations are acyclic, the admissible ranking set and every menu-level identified choice set weakly shrink.
\end{proposition}

\begin{proof}
By Lemma~\ref{lem:adr-char}, a compatible full-horizon ranking must extend $\POH$, so imposing the comparisons in $R$ is equivalent to asking for a strict linear order extending $Q:=\POH\cup R$. Necessity of acyclicity is immediate because no strict linear order can contain a directed cycle. Conversely, if $Q$ is acyclic, then its transitive closure $P_R=\tc(Q)$ is irreflexive and transitive, hence a strict partial order. A finite strict partial order has a strict linear extension, so a compatible ranking exists. Moreover, a strict linear order extends $Q$ if and only if it extends $\tc(Q)$, which proves that the admissible rankings are exactly $\LE(P_R)$. The menu statement then follows from Lemma~\ref{lem:extensions}(ii). Finally, if $R\subseteq\widetilde R$, then $\tc(\POH\cup R)\subseteq\tc(\POH\cup\widetilde R)$. Any linear extension of the larger relation is therefore an extension of the smaller one, and any element maximal under the larger relation is also maximal under the smaller relation. This proves both monotonicity claims.
\end{proof}

Incomplete full-horizon observations are a special case. This is the paper's constrained analogue of the classical revealed-preference idea that finite choice data are rationalizable by an ordering when the revealed strict comparisons are cycle-free \parencite{Richter1966}. Here the protected comparisons generated by the model enter the revealed relation as additional maintained restrictions. Suppose
\[
 \mathcal D^F=\{(S_t,f_t)\}_{t=1}^T
\]
is a dataset in which $f_t\in S_t$ is the observed full-horizon choice. Let
\[
 R(\mathcal D^F)
 =\{(f_t,y):t=1,\ldots,T,\ y\in S_t\setminus\{f_t\}\}
\]
collect the directly revealed pairwise comparisons.

\begin{corollary}[Incomplete-data rationalizability]\label{cor:data}
Fix an ordinary ranking $\succO$ and a Hereafter ranking $\succeqH$. The dataset $\mathcal D^F$ is rationalizable by a path-independent full-horizon choice function satisfying Hereafter-Directed Revision if and only if
\[
 \POH\cup R(\mathcal D^F)
\]
is acyclic. If so, the rationalizing full-horizon rankings are exactly the linear extensions of
\[
 \tc\bigl(\POH\cup R(\mathcal D^F)\bigr),
\]
and the maximal elements of this relation give the sharp set of model-consistent choices on every unobserved menu.

If the full two-horizon model is imposed and the Now ranking $\succeqN$ is also fixed, then the pair consisting of the ordinary ranking $\succO$ and the dataset $\mathcal D^F$ is rationalizable by path-independent ordinary and full-horizon choice satisfying both substantive axioms if and only if
\[
 \succO\in\LE(\PNH)
 \qquad\text{and}\qquad
 \POH\cup R(\mathcal D^F)\text{ is acyclic}.
\]
\end{corollary}

\begin{proof}
Apply Proposition~\ref{thm:augment} with $R=R(\mathcal D^F)$. A strict ranking rationalizes each full-horizon observation exactly when it places $f_t$ above every rejected alternative in $S_t$. This proves the first statement. For the full model, Theorem~\ref{thm:main} adds exactly the independent requirement that the fixed ordinary ranking extend $\PNH$; the full-horizon part is unchanged.
\end{proof}

Failure of the acyclicity test rejects the maintained restrictions jointly; it does not identify which component is misspecified.

\section{Concluding remarks}\label{sec:conclusion}

This paper develops a revealed-preference model of directional revision under two-dimensional deliberation. Two fixed consequence rankings are separated from the overall rankings revealed in two choice modes. A one-sided restriction on cross-mode reversals generates a protected partial order whose linear extensions are exactly the admissible full-horizon rankings. The characterization yields sharp identification, a signed menu-choice wedge, and finite rationalizability tests. Across ordered deliberative modes, actual choices form monotone tradeoff trajectories: every switch moves in the designated direction, abandoned alternatives cannot reappear, and the number of switches is bounded by the height of the induced tradeoff order. With a strict emphasized consequence order, the same restrictions admit the Kemeny--Kendall interpretation developed above.

The applications discussed in the introduction, including the Islamic-economic illustration, are formally on equal footing. What the framework requires is that two ordinal consequence rankings can reasonably be held fixed while an intervention has a defensible directional interpretation. It does not determine the threshold at which an improvement on one dimension should outweigh a loss on the other, nor does it cover interventions that alter beliefs, feasibility, factual information, or perceived consequences.

Natural extensions include incomplete or state-dependent consequence rankings, stochastic choice, and settings in which the ordering of deliberative modes must itself be inferred from data. The present results provide an ordinal benchmark: directional deliberation can generate substantial and testable restrictions on observed choice without a target utility function or a complete rule for resolving every cross-dimensional tradeoff.

\section*{Statements and Declarations}
\noindent\textbf{Competing interests.} The author declares no competing interests.\par
\noindent\textbf{Data availability.} No data were generated or analyzed in this study; the paper is entirely theoretical.\par

\printbibliography

@book{saheehinternational1997,
  author    = {{Saheeh International}},
  title     = {The Qur'an: Arabic Text with Corresponding English Meanings},
  year      = {1997},
  publisher = {Abul-Qasim Publishing House},
  location  = {Jeddah},
  isbn      = {9960792633},
  langid    = {english},
  note      = {English translation of the meanings of the Qur'an}
}

@article{BernheimRangel2009,
  author  = {Bernheim, B. Douglas and Rangel, Antonio},
  title   = {Beyond Revealed Preference: Choice-Theoretic Foundations for Behavioral Welfare Economics},
  journal = {Quarterly Journal of Economics},
  year    = {2009},
  volume  = {124},
  number  = {1},
  pages   = {51--104},
  doi     = {10.1162/qjec.2009.124.1.51}
}

@article{GulPesendorfer2001,
  author  = {Gul, Faruk and Pesendorfer, Wolfgang},
  title   = {Temptation and Self-Control},
  journal = {Econometrica},
  year    = {2001},
  volume  = {69},
  number  = {6},
  pages   = {1403--1435},
  doi     = {10.1111/1468-0262.00252}
}

@article{Plott1973,
  author  = {Plott, Charles R.},
  title   = {Path Independence, Rationality, and Social Choice},
  journal = {Econometrica},
  year    = {1973},
  volume  = {41},
  number  = {6},
  pages   = {1075--1091},
  doi     = {10.2307/1914037}
}

@article{SilvaLeanderSeth2017,
  author  = {Silva-Leander, Sebastian and Seth, Suman},
  title   = {Revealed Preferences with Plural Motives: Axiomatic Foundations of Normative Assessments in Non-Utilitarian Welfare Economics},
  journal = {Social Choice and Welfare},
  year    = {2017},
  volume  = {48},
  number  = {3},
  pages   = {505--517},
  doi     = {10.1007/s00355-016-1020-x}
}

@book{KemenySnell1962,
  author    = {Kemeny, John G. and Snell, J. Laurie},
  title     = {Mathematical Models in the Social Sciences},
  publisher = {Blaisdell Publishing Company},
  address   = {Waltham, MA},
  year      = {1962}
}

@article{Grandmont1978,
  author  = {Grandmont, Jean-Michel},
  title   = {Intermediate Preferences and the Majority Rule},
  journal = {Econometrica},
  year    = {1978},
  volume  = {46},
  number  = {2},
  pages   = {317--330},
  doi     = {10.2307/1913903}
}

@article{BossertSprumont2014,
  author  = {Bossert, Walter and Sprumont, Yves},
  title   = {Strategy-Proof Preference Aggregation: Possibilities and Characterizations},
  journal = {Games and Economic Behavior},
  year    = {2014},
  volume  = {85},
  pages   = {109--126},
  doi     = {10.1016/j.geb.2014.01.015}
}

@inproceedings{HaretWallner2022,
  author    = {Haret, Adrian and Wallner, Johannes Peter},
  title     = {An Axiomatic Approach to Revising Preferences},
  booktitle = {Proceedings of the AAAI Conference on Artificial Intelligence},
  year      = {2022},
  volume    = {36},
  number    = {5},
  pages     = {5676--5683},
  doi       = {10.1609/aaai.v36i5.20509}
}

@article{Freund2005,
  author  = {Freund, Michael},
  title   = {Revising Preferences and Choices},
  journal = {Journal of Mathematical Economics},
  year    = {2005},
  volume  = {41},
  number  = {3},
  pages   = {229--251},
  doi     = {10.1016/j.jmateco.2003.11.007}
}

@article{Spohn2025Reflective,
  author  = {Spohn, Wolfgang},
  title   = {Reflective Choice: A Novel Decision Rule for Changing Decision Situations},
  journal = {Theory and Decision},
  year    = {2025},
  doi     = {10.1007/s11238-025-10057-9}
}

@article{GruneYanoff2025Constructed,
  author  = {Grüne-Yanoff, Till},
  title   = {Intending Now or Later: The Rationality of Constructed Preferences},
  journal = {Theory and Decision},
  year    = {2025},
  doi     = {10.1007/s11238-025-10045-z}
}

@article{Petersen2026Dynamic,
  author  = {Petersen, Sami},
  title   = {Dynamic Choice without the Completeness Axiom},
  journal = {Theory and Decision},
  year    = {2026},
  doi     = {10.1007/s11238-025-10052-0}
}

@article{BreigFeldman2026RP,
  author  = {Breig, Zachary and Feldman, Paul},
  title   = {Revealed Preference Tests for Linear Probability--Prize Tradeoffs},
  journal = {Theory and Decision},
  year    = {2026},
  volume  = {100},
  number  = {2},
  pages   = {403--425},
  doi     = {10.1007/s11238-025-10089-1}
}

@article{BreigFeldman2024Revisions,
  author  = {Breig, Zachary and Feldman, Paul},
  title   = {Revealing Risky Mistakes through Revisions},
  journal = {Journal of Risk and Uncertainty},
  year    = {2024},
  volume  = {68},
  pages   = {227--254},
  doi     = {10.1007/s11166-024-09429-3}
}

@article{DushnikMiller1941,
  author  = {Dushnik, Ben and Miller, E. W.},
  title   = {Partially Ordered Sets},
  journal = {American Journal of Mathematics},
  year    = {1941},
  volume  = {63},
  number  = {3},
  pages   = {600--610},
  doi     = {10.2307/2371374}
}

@article{Kendall1938,
  author  = {Kendall, M. G.},
  title   = {A New Measure of Rank Correlation},
  journal = {Biometrika},
  year    = {1938},
  volume  = {30},
  number  = {1/2},
  pages   = {81--93},
  doi     = {10.1093/biomet/30.1-2.81}
}

@article{Richter1966,
  author  = {Richter, Marcel K.},
  title   = {Revealed Preference Theory},
  journal = {Econometrica},
  year    = {1966},
  volume  = {34},
  number  = {3},
  pages   = {635--645},
  doi     = {10.2307/1909773}
}

\end{document}